\documentclass{article}
\usepackage{amsmath,amsfonts,amssymb,amsthm,booktabs,color,epsfig,graphicx,url}
\usepackage{txfonts}
\usepackage[bookmarkstype=toc, colorlinks=true, linkcolor=cyan, pdfborder={0 0 0}, bookmarks=true, bookmarksnumbered=true]{hyperref}

\theoremstyle{definition}

\newtheorem{proposition}{Proposition}

\newtheorem{definition}{Definition}

\newtheorem{example}{Example}
\usepackage{algorithmic}
\usepackage{algorithm}

\usepackage{graphicx} 
\graphicspath{ {results} }
\usepackage[labelfont=bf]{caption}

\usepackage[affil-it]{authblk}
\title{Relative local dependence of bivariate copulas}
\author{Issey Sukeda}
\affil{Department of Mathematical Informatics, Graduate School of Information Science and Technology, The University of Tokyo; \\RIKEN Center for Brain Science}

\author{Tomonari Sei}
\affil{Graduate School of Information Science and Technology, The University of Tokyo}
\date{}

\begin{document}

\maketitle

\begin{abstract}
For a bivariate probability distribution, local dependence around a single point on the support is often formulated as the second derivative of the logarithm of the probability density function. However, this definition lacks the invariance under marginal distribution transformations, which is often required as a criterion for dependence measures. In this study, we examine the \textit{relative local dependence}, which we define as the ratio of the local dependence to the probability density function, for copulas. By using this notion, we point out that typical copulas can be characterised as the solutions to the corresponding partial differential equations, particularly highlighting that the relative local dependence of the Frank copula remains constant. The estimation and visualization of the relative local dependence are demonstrated using simulation data. Furthermore, we propose a class of copulas where local dependence is proportional to the $k$-th power of the probability density function, and as an example, we demonstrate a newly discovered relationship derived from the density functions of two representative copulas, the Frank copula and the Farlie-Gumbel-Morgenstern(FGM) copula.
\end{abstract}

\section{Introduction \label{sec:introduction}}

Dependence is a concept of importance in statistics, and has been extensively studied in tremendous papers. 
Marginal-invariance has been critical for dependence measures. For global dependence measures, while Pearson's correlation is not marginal-invariant, Spearman's $\rho$ and Kendall's $\tau$ have been widely used. It is known that these measures are expressed merely via copulas:
$\rho = \int_0^1 \int_0^1 12uvdC(u,v) - 3,$
and $\tau = 4\int_0^1 \int_0^1 C(u,v)dC(u,v) - 1,$
where $C$ is the cumulative distribution function (cdf) of the responding copula.

While many traditional dependence measures quantify global dependence, some studies have focused on the local dependence.
A typical way of defining local dependence of a distribution having the probability density function (pdf) $f(x,y)$ is 
$\frac{\partial^2}{\partial x \partial y}\log{f(x,y)},$
which was first proposed by Holland and Wang~\cite{holland1987}, further studied by Jones~\cite{JONES1998148,jones2003}, and used in various works~\cite{molenberghs1997nonlinear,sankaran2004,gupta2010local,koutoumanou2017local}.
Kurowicka and van Horssen~\cite{KUROWICKA2015127} studied it on copulas in the name of \textit{interaction function}. This notion is rather intuitive, as being the continuous version of log odds ratio, which is a notion of positive dependence on contingency tables~\cite{holland1987}. 
While this measure is convenient, we point out that it is not marginal-invariant and not preferable as a dependence measure.
To overcome this problem, we introduce a marginal-invariant local dependence measure defined as the ratio of the density function to its local dependence
$$\frac{1}{f(x,y)}\frac{\partial^2}{\partial x \partial y}\log{f(x,y)}$$
and name it \textit{relative local dependence}.

The following paper is structured as follows. 
In Section 2, we review the basic of copulas and the local dependence function. Section 3 is devoted to our new proposal, the use of \textit{relative local dependence}. The relative local dependence of typical copulas are presented, followed by the relationship between local Kendall's $\tau$. The estimation and visualization of relative local dependence are also demonstrated using simulation data in Section 4. 
Section 5 is dedicated to the extension of Section 3, which leads to the discovery of a new relationship between the well-known Frank copula and the FGM copula. 
Finally in Section 6, we discuss the limitations and the future works, and conclude our work.

\section{Preliminaries}
\subsection{Copulas}
A copula is roughly a probability distribution defined on $[0,1]^2$ with all marginals being uniform distributions~\cite{nelsen2007introduction}. Although copulas are extended to higher dimensions as well, we only consider bivariate absolute continuous copulas for simplicity in this paper. 
It is widely known that any joint density can be decomposed into marginal densities and a copula density from Sklar's theorem. In a bivariate case, 
\begin{align}
f(x_1,x_2) = f_{1}(x_1)f_2(x_2)c(F_1(x_1),F_2(x_2))  \label{eq:sklar}
\end{align}
where $F_1, F_2$ are the marginal distributions, $f_1, f_2$ are their densities, and $c$ is the copula density, respectively. 
Inversely, the copula density is determined from the joint density and its marginals:
$$c(u_1,u_2) = \frac{f(F_1^{-1}(u_1),F_2^{-1}(u_2))}{f_1(F_1^{-1}(u_1))f_2(F_2^{-1}(u_2))}.$$
When \eqref{eq:sklar} holds, it means that the distribution function $F$ (or $f$) becomes a copula function $C$ (or $c$) after the \textit{probability integral transform} $U = F(X)$ applied to $x_1$ and $x_2$ independently. We express this relationship as ``The distirbution $F$ has a copula $C$'' in this paper.

\subsection{Local dependence function}

The \textit{local dependence function}~\cite{holland1987} (or \textit{local dependence}~\cite{KUROWICKA2015127}) of $f$ is defined as follows:
\begin{definition}[local dependence] The local dependence function for $(x_1,x_2)$ with twice differentiable joint density $f$ is defined as 
$$i^f(x_1,x_2) = \frac{\partial^2}{\partial x_1 \partial x_2} \log{f(x_1,x_2)}.$$
\end{definition}
\noindent From \eqref{eq:sklar}, the following relationship between local dependence functions holds:
$$i^f(x_1,x_2) =\frac{\partial^2}{\partial x \partial y} \log{f(x_1,x_2)} = i^c(F_1(x_1),F_2(x_2))f_1(x_1)f_2(x_2).$$
$$i^c(u_1,u_2) = \frac{\partial^2}{\partial u_1 \partial u_2} \log{c(u_1,u_2)} = \frac{i^f(F_1^{-1}(u_1), F_2^{-1}(u_2))}{f_1(F_1^{-1}(u_1))f_2(F_2^{-1}(u_2))} $$
\noindent This notion of positive dependence has been studied in many literatures~\cite{GUPTA20101267, holland1987, JONES1998148, jones2003, KUROWICKA2015127}.
It also can be rewritten as 
$$i^c(u_1,u_2) = \frac{c^{11}(u_1,u_2)}{c(u_1,u_2)} - \frac{c^{10}(u_1,u_2)c^{01}(u_1,u_2)}{\{c(u_1,u_2)\}^2},$$
where $c^{11}(u_1,u_2) = \frac{\partial^2}{\partial u_1 \partial u_2} c(u_1,u_2), c^{10}(u_1,u_2) = \frac{\partial}{\partial u_1} c(u_1,u_2), c^{01}(u_1,u_2) = \frac{\partial}{\partial u_2} c(u_1,u_2)$.

\begin{example}[Gaussian distribution and Gaussian copula]
    The local dependence function for a bivariate Gaussian density with standard normal margins and correlation $\rho$ is $\frac{\rho}{1-\rho^2}$. On the other hand, the local dependence function of its copula $c^{Gaussian}_\rho(u_1,u_2)$ is $\frac{1}{\phi(\Phi^{-1}(u_1))\phi(\Phi^{-1}(u_2))} \frac{\rho}{1-\rho^2}$ where $\phi$ and $\Phi$ denotes pdf and cdf of the standard normal distribution, respectively.
\end{example}

While this definition of local dependence is widely known, it lacks the invariance property to the monotone transformation in marginal distributions. While this property is not necessarily required for the measure of the dependence, it is considered natural and  preferable. As Renyi~\cite{renyi1959measures} states, ``the following sets of postulates for an appropriate measure of dependence, which shall be denoted as $\delta(\xi,\eta)$, seems to be natural: \dots (F) If the Borel-measurable functions $f(x)$ and $g(x)$ maps the real axis in a one-to-one way on iteslf, $\delta(f(\xi),g(\eta)) = g(\xi,\eta)$. ''

\section{Relative local dependence}
Considering the issue previously pointed out, we propose to consider the modified version of local dependence function instead, which we name as the \textit{relative local dependence}.  
We define the relative local dependence as the local dependence function divided by the original density:
$$r^f(x_1,x_2) = \frac{1}{f(x_1,x_2)}\frac{\partial^2}{\partial x_1 \partial x_2} \log{f(x_1,x_2)}. \label{eq:ri}$$
\noindent The relative local dependence is preferable to the ordinary local dependence as a notion of dependence because it does not reflect the influence of marginal distributions. In fact, for any joint density $f$, its relative local dependence is determined only by the copula density of $f$.

\begin{proposition}\label{prop:invariance}
    Let $f(x,y)$ a joint density and $c(u,v)$ a copula density of $f$, where $u = F_1(x), v=F_2(y)$. Then, $r^f(x,y) = r^c(u,v)$.
\end{proposition}
\begin{proof}
The joint density $f$ is written as 
$$f(x,y) = f_1(x)f_2(y)c(F_1(x),F_2(y)),$$
where $f_1$ and $f_2$ are marginal densities of $x$ and $y$, respectively. Hence, 
\begin{align*}
    r^f(x,y) &= \frac{1}{f(x,y)}\frac{\partial^2}{\partial x \partial y} \log{f(x,y)}\\
    &= \frac{1}{f_1(x)f_2(y)c(F_1(x),F_2(y))}\frac{\partial^2}{\partial x \partial y} \log{f_1(x)f_2(y)c(F_1(x),F_2(y))}\\
    &= \frac{1}{f_1(x)f_2(y)c(F_1(x),F_2(y))}\frac{\partial^2}{\partial x \partial y} \log{c(F_1(x),F_2(y))}\\
    &= \frac{1}{c(u,v)}\frac{\partial^2}{\partial u \partial v} \log{c(u,v)}\\
    &= r^c(u,v)
\end{align*}
\end{proof}

For some typical bivariate copulas, we found that the relative local dependence can be represented in a simple form. From Section 3.1 to Section 3.3, we view these relationships between the relative local dependence and the pdf/cdf of copulas, which will be summarized in Table~\ref{tab:local dependences}. Section 3.4 reveals the relationship between the relative local dependence and Kendall's $\tau$, followed by application to visualizing local dependence via dependence map~\cite{jones2003} in Section 3.5.

\subsection{Relative local dependence of Archimedean copulas}

Let us consider a strict bivariate Archimedean copula
$$C^{\mathrm{A}}(u, v) = \psi(\psi^{-1}(u) + \psi^{-1}(v)),$$
where $\psi(0) = 1$, $\psi(\infty) = 0$, $(-1)^k \psi^{(k)}(t) > 0$ for $k=1,2$. $\psi$ is called a ``\textit{generator}''.
The density function of Archimedean copula is represented as
$$c^\mathrm{A}(u,v) = \frac{\partial^2}{\partial u \partial v}C^\mathrm{A}(u,v)= \frac{\psi''(\psi^{-1}(u) + \psi^{-1}(v))}{ \psi'(\psi^{-1}(u))\psi'(\psi^{-1}(v))}.$$

\noindent Assume that the fourth derivative of the
generator, denoted as $\psi^{(4)}$, exists. The local dependence and the relative local dependence of Archimedean copula are
$$i^\mathrm{A}(u, v) = \frac{1}{\psi'(\psi^{-1}(u))\psi'(\psi^{-1}(v))}\frac{d^2}{dt^2}[\log{\psi''(t)}]_{t=\psi^{-1}(u)+\psi^{-1}(u)}$$
and
$$r^A(u, v) = \left. \frac{1}{\psi''(t)}\frac{d^2}{dt^2}\log{\psi''(t)}\ \right|_{t=\psi^{-1}(u)+\psi^{-1}(v)},$$
respectively. Here we calculate the relative local dependence for four typical Archimedean copulas: Frank, Clayton, Gumbel-Hougaard, and Ali-Mikhail-Haq. The different features depending on the type of copulas can be observed in Figure~\ref{fig:diagonal-rdf}, where the value of the relative local dependence on the diagonal ($u = v \in [0,1]$) of each copula is depicted. The examples presented below are the explicit form of the relative local dependence of Frank copula and Clayton copula. 

\begin{figure}
    \centering
    \includegraphics[width=\linewidth]{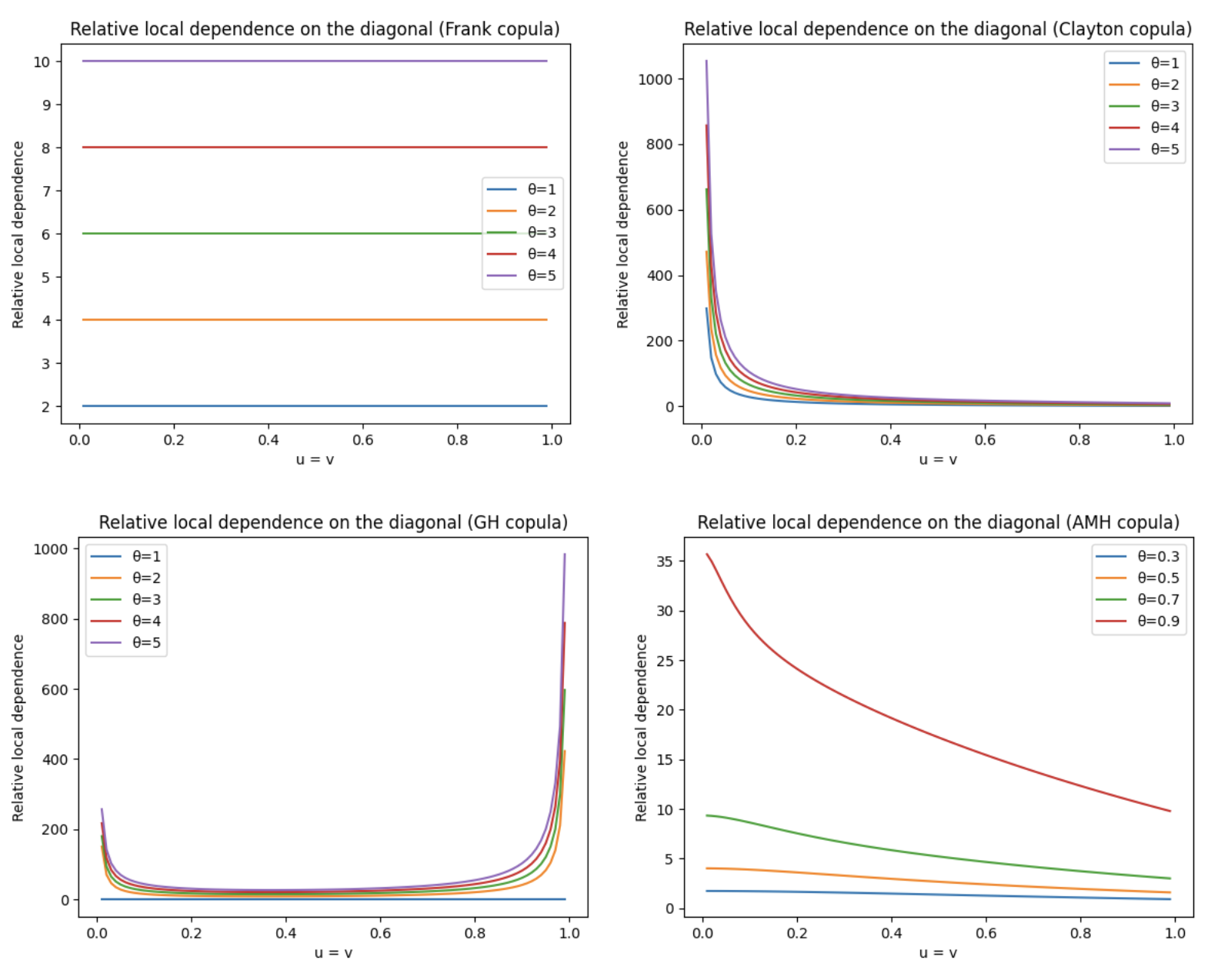}
  \caption{The relative local dependence along the diagonal of Archimedean copulas. $u=v=0$ indicates the left bottom tail and  $u=v=1$ indicates the right upper tail.}
  \label{fig:diagonal-rdf}
\end{figure}

\begin{example}[Frank copula]\label{example:frank}
The cdf of Frank copula is defined as
$$C_\theta^{Frank}(u,v) = \psi_\theta(\psi_\theta^{-1}(u) + \psi_\theta^{-1}(v)) = -\frac{1}{\theta} \ln \left(1+\frac{(e^{-\theta u}-1)(e^{-\theta v}-1)}{e^{-\theta}-1}\right), \ \theta \neq 0,$$
where
$$\psi_\theta(t) = -\frac{1}{\theta}\ln{\left(1-(1-e^{-\theta})e^{-t} \right)},\ \psi^{-1}_\theta(s) = -\ln{\left( \frac{e^{-\theta s} - 1}{e^{-\theta }-1}\right)}.$$
\end{example}
By calculating derivatives, Frank density is 
$$c_\theta^{Frank}(u, v) = \frac{\theta (1 - e^{-\theta})e^{-\theta(u + v)}}{\{1 - e^{-\theta} - (1 - e^{-\theta u})(1 - e^{-\theta v})\}^2}.$$
The direct calculation leads to
$$i_\theta^{{Frank}}(u, v) = 2\theta c_\theta^{Frank}(u, v),$$
and
$$r_\theta^{{Frank}}(u, v) = 2\theta.$$

\begin{example}[Clayton copula]
The cdf of Clayton copula is defined as 
$$C_\theta^{Clayton}(u,v)  = \psi_\theta(\psi_\theta^{-1}(u) + \psi_\theta^{-1}(v)) = \frac{1}{(u^{-\theta}+v^{-\theta}-1)^{\frac{1}{\theta}}}, \theta \in [1,\infty)$$
where
$$\psi_\theta(t) = (1+t)^{-\frac{1}{\theta}},\  \psi_\theta^{-1}(s) = s^{-\theta} - 1 .$$
Clayton density is calculated as
$$c_\theta^{Clayton}(u,v) = (1+\theta)\frac{1}{u^{1+\theta}v^{1+\theta}}\frac{1}{(u^{-\theta}+v^{-\theta}-1)^{\frac{1}{\theta}+2}}.$$


The direct calculation leads to
$$i_\theta^{\ \mathrm{Clayton}}(u,v) = \theta(1+2\theta)\frac{1}{u^{1+\theta}v^{1+\theta}}\frac{1}{(u^{-\theta}+v^{-\theta}-1)^2} = c_\theta^{Clayton}(u,v)\frac{\theta(1+2\theta)}{1+\theta}(u^{-\theta}+v^{-\theta}-1)^{\frac{1}{\theta}}= \frac{\theta(1+2\theta)}{1+\theta}\frac{c_\theta^{Clayton}(u,v)}{C_\theta^{Clayton}(u,v)},$$
and
$$r_\theta^{{Clayton}}(u, v) = \frac{\theta(1+2\theta)}{1+\theta}\frac{1}{C_\theta^{Clayton}(u,v)}.$$
\noindent It is observed that the relative local dependence function of Clayton copula is proportional to the inverse of its cdf. This is consistent with the fact that Clayton copula has heavier tail dependence near $(0,0)$ and lighter tail dependence near $(1,1)$.
\end{example}

\subsection{Relative local dependence of the minimum information copula}
In Example~\ref{example:frank}, it is observed that the relative local dependence of Frank copulas is a constant, depending on the parameter of its copula density. In contrast, the distribution with a constant local dependence function has been studied previously by Jones~\cite{JONES1998148}. The copula density with a constant local dependence is written as
$$p(u,v) = a(u)b(v)\exp{(\theta uv)}$$
where $\theta$ is an arbitrary constant, and $a(u)$ and $b(v)$ are normalizing functions such that
$$\int_0^1 a(u)b(v)\exp{(\theta uv)} du = 1$$
and 
$$\int_0^1 a(u)b(v)\exp{(\theta uv)} dv = 1.$$

Here we point out that this distribution is identical to the minimum information copula that has been studied in different literatures and contexts. This copula density was first proposed by Meeuwissen and Bedford~\cite{MEEU1997} and was generalized by Bedford and Wilson~\cite{bedford2014construction}. The minimum information copula under fixed Spearman's $\rho$ (MICS) is defined as the optimal solution of the following problem:
\begin{equation}
\mathrm{minimize}\ \int_0^1 \int_0^1 p(u,v)\log{p(u,v)} \mathrm{d}u\mathrm{d}v, \label{prob:mick}
\end{equation}
$$\mathrm{s.t.}\ \int_0^1 p(u,v) \mathrm{d}u = 1,\ \int_0^1 p(u,v) \mathrm{d}v = 1,$$
$$0 \leq p(u,v),$$
$$\rho[p] = \int_0^1 \int_0^1 12(u-1/2)(v-1/2)p(\tilde{u},\tilde{v}) \mathrm{d}u\mathrm{d}v = \mu.$$
By employing the Lagrangian method, The optimal solution of this problem takes the form $p(u,v) = a(u)b(v)\exp{(\theta uv)}$. 
The local dependence function is immediate for this copula:
$$i_{\theta}^{MICS}(u,v) = \theta.$$

\subsection{Relative local dependence of Rodriguez-Lallena and Ubeda-Flores family of copula (RU copula)}

RU copula is proposed in ~\cite{rodriguez2004new} originally, and also known as ``maximum Tsallis entropy copula''~\cite{pougaza2012new}. Its cumulative distribution function and density are given as
$$C_\theta^{RU}(u,v) = uv + \theta A(u)B(v)$$
and
$$c_\theta^{RU}(u,v) = 1+\theta a(u)b(v)$$
respectively, where $A(u)$ and $B(v)$ are arbitrary functions satisfying $A(0)=A(1)=B(0)=B(1)=0$, $a(u) = A'(u)$, and $b(v) = B'(v)$. For this copula family, its local dependence is easily calculated as 
$$i_\theta^{{RU}}(u,v) = \theta a'(u)b'(v)(c_\theta^{RU}(u,v))^{-2}.$$
\noindent As a representative example, RU family includes Farlie-Gumbel-Morgenstern copula as the case where $a(u) = 2u-1$ and $b(v)=2v-1$.
\begin{example}[Farlie-Gumbel-Morgenstern (FGM) copula]
The FGM copula is defined as
$$C_\theta^{FGM}(u,v)=uv+\theta uv(1-u)(1-v), \theta \in [-1,1].$$
The density function is given by
$$c_\theta^{FGM}(u,v) = 1+\theta(2u-1)(2v-1).$$
Hence, 
$$i_\theta^{{FGM}}(u,v) = 4\theta (c_\theta^{FGM}(u,v))^{-2},$$
$$r_\theta^{{FGM}}(u,v) = 4\theta (c_\theta^{FGM}(u,v))^{-3}.$$
The relative local dependence along the diagonal is depicted in Figure~\ref{fig:fgm-diagonal-rld}, peaking at the center $u = v = 0.5$ and approaching $\frac{4\theta}{(1+\theta)^3}$ towards the both corners. 

\begin{figure}
    \centering
    \includegraphics[width=0.5\linewidth]{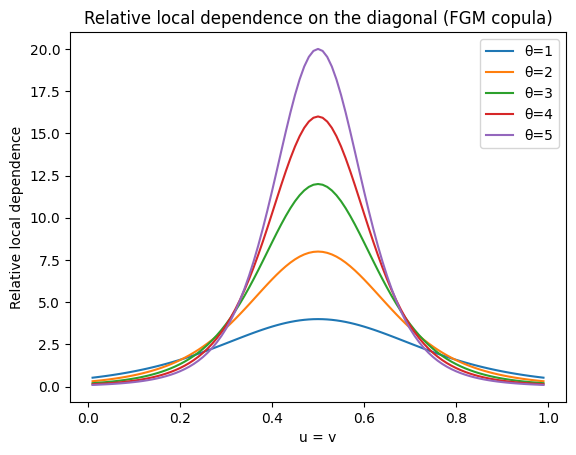}
    \caption{The relative local dependence along the diagonal of tje FGM copula. The function can be explicitly written as $r^{FGM}(u,u) = \frac{4\theta}{\{1+\theta(2u-1)\}^3}$.}
    \label{fig:fgm-diagonal-rld}
\end{figure}
\end{example}

\begin{table}[t!]
    \centering
    \begin{tabular}{|c|c|c|c|c|}\hline
        name of copula  & is Archimedean? &local dependence & relative local dependence \\ \hline \hline
        Uniform & yes & $= 0$&$= 0$\\
        MICS & no &$=$ constant&$\propto c(u,v)^{-1}$\\
        Frank & yes &$\propto c(u,v)$&$=$ constant\\
        Clayton & yes &$\propto c(u,v)/C(u,v)$&$\propto C(u,v)^{-1}$\\
        FGM& no &$\propto c(u,v)^{-2}$ & $\propto c(u,v)^{-3}$\\ \hline
    \end{tabular}
    \caption{Local dependence and relative local dependence of important classes of copulas}
    \label{tab:local dependences}
\end{table}

\subsection{Relationship between relative local dependence and local Kendall's $\tau$}

(Global) Kendall's $\tau$ is one of the most common notions of positive dependence ubiquitous in many researches and data analysis. Kendall's $\tau$ is independent of marginal distributions and thus defined solely by the copula function:
$$\tau = 4\int_0^1 \int_0^1 C(u,v)dC(u,v) - 1 = \int_{0}^{1} \int_{0}^{1}\int_{0}^{1}\int_{0}^{1} \mathrm{sgn}(u-\tilde{u})\mathrm{sgn}(v-\tilde{v})c(u,v)c(\tilde{u},\tilde{v})dudvd\tilde{u}d\tilde{v}.$$
On the other hand, extensions of the definition of Kendall's $\tau$ has been studied in several literature. Recently, four types of local Kendall's $\tau$ were proposed by Huang~\cite{huang2018copula}, which have been utilized by Albulescu et al.~\cite{ALBULESCU2020117762} to analyze dependence structure of commodity markets. 
One of them is defined on the bottom left corner as 
\begin{equation} \label{eq:tau_LL}
\tau_{LL}(X,Y; p,q) = \frac{4\int_0^p\int_0^q C(u,v)dC(u,v)}{C(p,q)^2}-1, 
\end{equation}
which coincides with Venter~\cite{venter2002tails}'s cumulative Kendall's $\tau$ when $p = q$. When $(p,q)=(1,1)$, it reduces to the normal $\tau$. Another representation of $\tau_{LL}$ is derived as follows. See \ref{appendix:local-kendall-tau} for its derivation.
\begin{proposition}\label{prop:tau_LL}
    $$\tau_{LL}(X,Y; p,q) = \frac{\int_{0}^{p} \int_{0}^{q}\int_{0}^{p}\int_{0}^{q} \mathrm{sgn}(u-\tilde{u})\mathrm{sgn}(v-\tilde{v})c(u,v)c(\tilde{u},\tilde{v})dudvd\tilde{u}d\tilde{v}}{C(p,q)^2}.$$
\end{proposition}
\noindent In contrast to the usual lower tail dependence $\lambda_{L} = \lim_{p\to +0}\frac{C(p,p)}{p}$, the novel lower tail dependence was defined in the same literature as the limit of local Kendall's $\tau$ on the diagonal\footnote{The existence of the limiting value is not discussed by Huang~\cite{huang2018copula}.}:
$$\lambda_{LL}^{Kendall} = \lim_{p\to 0} \tau_{LL}(X,Y; p,p) \in [-1,1].$$

\begin{example}[Clayton copula $c_\theta^{Clayton}$]\footnote{This calculation is not consistent with Appendix A. of Venter~\cite{venter2002tails}. We further confirmed that $\tau_{LL}$ becomes constant for Clayton copulas using numerical integration to calculate \eqref{eq:tau_LL}.} $\tau_{LL}(X,Y; p,p) = \frac{\theta}{\theta+2},\  \lambda_{LL}^{Kendall} = \lim_{p\to 0} \tau_{LL}(X,Y; p,p) = \frac{\theta}{\theta+2}.$
    
\end{example}
\begin{example}[FGM copula $c_\theta^{FGM}$~\cite{huang2018copula}]
    $\tau_{LL}(X,Y; p,p) = \frac{2\theta p^2}{9(1+\theta(p-1))^2},\  \lambda_{LL}^{Kendall} = \lim_{p\to 0} \tau_{LL}(X,Y; p,p) = 0.$
\end{example}

Here we consider an extension of the local Kendall's $\tau$. While $\tau_{LL}$ is defined on a rectangular region $[0,p]\times[0,q]$, the naive extension to an arbitrary rectangular region can be made.
\begin{definition}[Local Kendall's $\tau$ for a general rectangular region (naive)] For a rectangular region $[p_1,p_2]\times[q_1,q_2]\ (0\leq p_1 < p_2 \leq 1, 0\leq q_1 < q_2 \leq 1)$, we define the local Kendall's $\tau$ as 
\begin{align*}
\tau_{[p_1,p_2]\times[q_1,q_2]}&= \frac{\int_{p_1}^{p_2} \int_{q_1}^{q_2}\int_{p_1}^{p_2}\int_{q_1}^{q_2} \mathrm{sgn}(u-\tilde{u})\mathrm{sgn}(v-\tilde{v})c(u,v)c(\tilde{u},\tilde{v})dudvd\tilde{u}d\tilde{v}}{\{C(p_1,q_1)+C(p_2,q_2)-C(p_1,q_2)-C(p_2,q_1)\}^2}.
\end{align*}
\end{definition}
\noindent This definition is consistent with the local Kendall's $\tau$ \eqref{eq:tau_LL} defined by Huang~\cite{huang2018copula}, i.e., $\tau_{[0,p]\times[0,q]} = \tau_{LL}(X,Y;p,q)$.

However, this naive definition extending $\tau_{LL}$ is problematic in measuring local dependence in extreme cases such as on a single inner point. 

\begin{proposition} Suppose $c$ is a twice differentiable copula density and $(p_1,q_1)$ is an inner point of $[0,1]^2$. Then, if the limit exists,
    $\tau_{[p_1,p_2]\times[q_1,q_2]}$ converges to 0 when $p_2 \to p_1$ and $q_2 \to q_1$.
\end{proposition}

\begin{proof}
    Let $p_2 = p_1+dp_1$ and $q_2 = q_1+dq_1$. 
    Use Taylor expansion of $c(u,v)$ around $(p_1,q_1)$. We obtain
    $$\tau_{[p_1,p_1+dp_1]\times[q_1,q_1+dq_1]} = \frac{2(\frac{\partial^2 c(p_1,q_1)}{\partial u \partial v}c(p_1,q_1)-\frac{\partial c(p_1,q_1)}{\partial u}\frac{\partial c(p_1,q_1)}{\partial v})(\frac{1}{6}dp_1^3)(\frac{1}{6}dq_1^3)}{(c(p_1,q_1)dp_1dq_1)^2} = \frac{1}{18}i^c(p_1,q_1)dp_1dq_1$$
\end{proof}

\noindent Considering this result, we propose to make a slight modification to the definition of the local Kendall's $\tau$. Here, the local Kendall's $\tau$ is associated with the relative local dependence.
\begin{definition}[Local Kendall's $\tau$ for a general rectangular region (modified)]
$$\tau^*_{[p_1,p_2]\times[q_1,q_2]} = \frac{\int_{p_1}^{p_2} \int_{q_1}^{q_2}\int_{p_1}^{p_2}\int_{q_1}^{q_2} \mathrm{sgn}(u-\tilde{u})\mathrm{sgn}(v-\tilde{v})c(u,v)c(\tilde{u},\tilde{v})dudvd\tilde{u}d\tilde{v}}{\{C(p_1,q_1)+C(p_2,q_2)-C(p_1,q_2)-C(p_2,q_1)\}^3}.$$
\end{definition}

\begin{proposition}
Suppose $c$ is a twice differentiable copula density and $(p_1,q_1)$ is an inner point of $[0,1]^2$.
For the modified local Kendall's $\tau$ for a general rectangular region, 
  $$\lim_{p_2 \searrow p_1\\ q_2 \searrow q_1} \tau^*_{[p_1,p_2]\times[q_1,q_2]} = \frac{1}{18}r^c(p_1,q_1)$$  
when the limit exists.
\end{proposition}
\begin{proof}
    Similar to the previous one, where the whole term is finally divided by $c(p_1,q_1)dp_1dq_1$. Note that $r^c(u,v) = \frac{i^c(u,v)}{c(u,v)}$ from the definition of relative local dependence.
\end{proof}

\section{Estimation and visualization of the relative local dependence}

While the local dependence changes with monotone transformation in marginals, the relative local dependence remains constant. Therefore, we believe estimating the relative local dependence is more essential for understanding the local dependence of the joint distribution behind bivariate data. In this section, we propose a direct estimator of the relative local dependence based on the Frank copula, which we name \textit{local Frank fitting}, and compare it to the naive estimation using Jones and Koch~\cite{jones2003}.

\subsection{Local bilinear fitting}

Firstly, we review the estimation proposed in Jones and Koch~\cite{jones2003} for local dependence. The kernel estimator they proposed takes the form 
\begin{equation}\label{eq:kernel-estimator}
\hat{i}(x,y) = \frac{g_{11}(x,y) - \frac{g_{01}(x,y)g_{10}(x,y)}{g_{00}(x,y)}}{h_1^2h_2^2s_2^2g_{00}(x,y)}, 
\end{equation}
where $g_{11}(x,y) = n^{-1}\sum_{i=1}^n X_iK_{h_1}(X_i-x)Y_iK_{h_2}(Y_i-y)$, $g_{10}(x,y) = n^{-1}\sum_{i=1}^n X_iK_{h_1}(X_i-x)K_{h_2}(Y_i-y)$, $g_{01}(x,y) = n^{-1}\sum_{i=1}^n K_{h_1}(X_i-x)Y_iK_{h_2}(Y_i-y)$, $g_{00}(x,y) = n^{-1}\sum_{i=1}^n K_{h_1}(X_i-x)K_{h_2}(Y_i-y)$, $s_2 = \int u^2 K(u)du$, and $K_h(\cdot) = h^{-1}K(h^{-1}\cdot)$.
The kernel function is typically the biweight univariate density:
$$K(u) = \frac{15}{16}(1-u^2)^2.$$
This estimator was derived by locally fitting a bilinear form $a+bx+cy+dxy$ to the log density. Here, $h_1$ and $h_2$ are bandwidth parameters controlling smoothness of the result in the $x$ and $y$ directions, respectively.

While the estimator $\hat{i}$ estimates the local dependence of the joint density, we notice that the term $g_{00}(x,y)$ in the denominator is the kernel density estimator for a bivariate joint density itself. Therefore, the naive estimator of the relative local dependence is obtained as 
$$\hat{r}(x,y) = \frac{g_{11}(x,y) - \frac{g_{01}(x,y)g_{10}(x,y)}{g_{00}(x,y)}}{h_1^2h_2^2s_2^2\{g_{00}(x,y)\}^2}.$$

The estimated results are shown in Figure~\ref{fig:frank3}. In this experiment, we sampled 10000 samples from the Frank copula with its parameter $\theta = 3$. Theoretically, the ground truth of its local dependence function is $6c_{\theta=3}^{Frank}$ and its relative local dependence is 6 everywhere on the support, which are visualized in the left-hand side of Figure~\ref{fig:frank3}. The estimation results of the density function, the local dependence function, and the relative local dependence are shown in the top, middle, bottom of the right-hand side of Figure~\ref{fig:frank3}, respectively.
Overall, it can be observed that the estimation performance is relatively good near the center, while it exacerbates towards the four corner of $[0,1]^2$. Specifically in estimating relative local dependence, the histogram shows that the estimated values are lower than the true value $6$ in the entire region. 
We shall discuss three causes of this issue. Firstly, this naive estimator can be biased. Secondly, this method suffers from the border bias problem as well as the other kernel-based estimation methods. Especially, while the Frank copula has the higher density near $(0,0)$ and $(1,1)$, the kernel estimator ($g_{00}(x,y)$ in \eqref{eq:kernel-estimator}) is not able to capture this feature. This insufficient capability of the density estimator is inherited to the relative local dependence estimator since it includes the division with $g_{00}(x,y)$. Last but not least, the sample points tend to lack near $(0,1)$ and $(1,0)$ especially under the case of strong dependence, which is likely to lead to inaccurate estimates. Indeed, Jones and Koch~\cite{jones2003} suggests to ignore regions with low density when visualising the dependence map.


\begin{figure}[htbp]
    \centering
    \includegraphics[width=\linewidth]{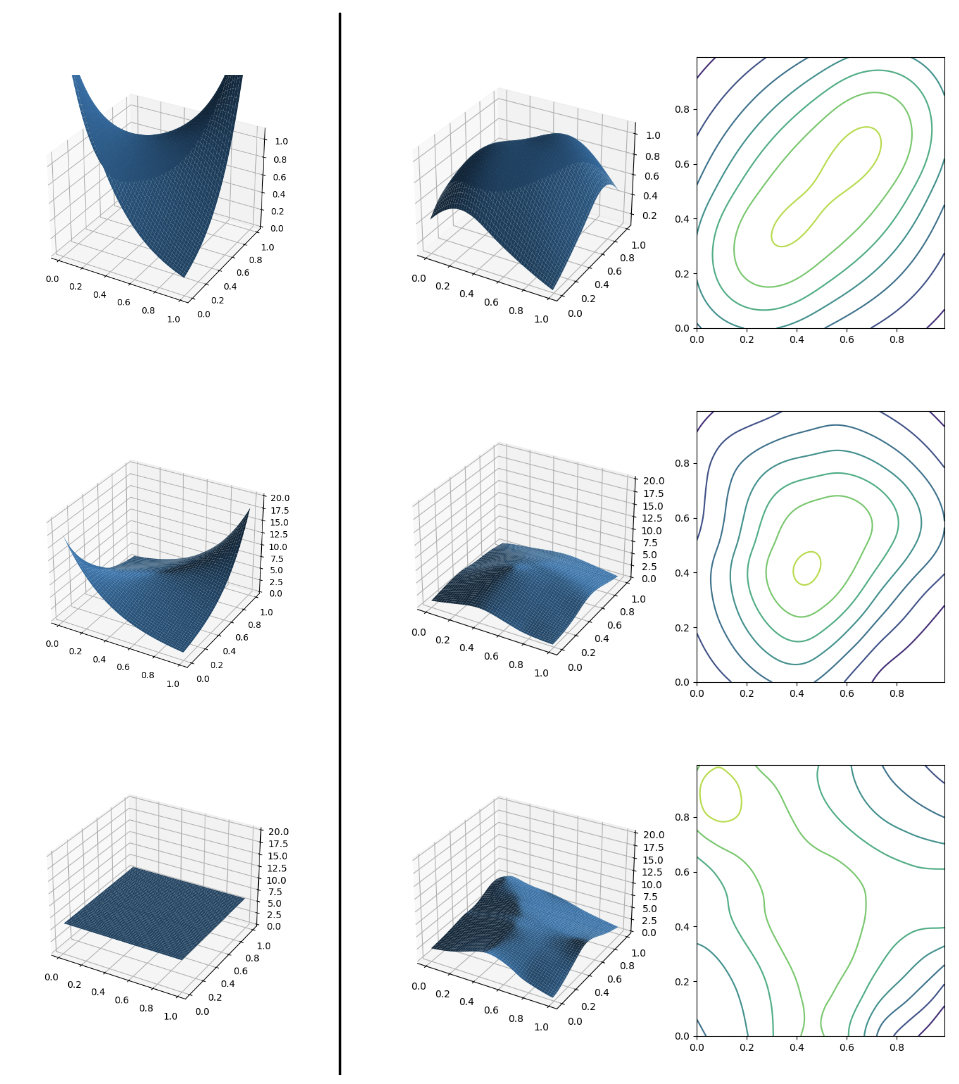}
  \caption{The kernel-based estimation results with simulated data sampled from $c_{\theta=3}^{Frank}$. (left: ground truth, right: estimated results / top: the density function, middle: the local dependence function, bottom: the relative local dependence)}
  \label{fig:frank3}
\end{figure}

On the other hand, (relative) local dependence of a single point $(x,y)$ entails difficulty in its interpretability. Here we further perform the visualization method called the \textit{dependence map}, following Jones and Koch~\cite{jones2003}. The dependence map is defined as the macro perspective of local dependence, which is intended to offer a middle ground between a scalar global metric that frequently smooths out too much of the dependence structure, and the comprehensive local dependence function, which is often overly detailed to be immediately practical. Specifically, each region is classified into three categories -- positive dependence, negative dependence, and neutral -- according to the hypothesis test that examines the null hypothesis $\hat{i} = 0$. Since $\hat{i} = 0$ and $\hat{r} = 0$ is truly equivalent to each other, the dependence map via local dependence and via our relative local dependence should look similar but slightly different due to the hypothesis testing procedures. 

Figure~\ref{fig:dependence-map-simulated} and Figure~\ref{fig:dependence-map-simulated2} show the dependence map of relative local dependence and local dependence function to simulated data respectively. 
This is done for 250 observations from the distribution of $(X, Y)$ where $X \sim \mathcal{N}(0,1)$ and $Y = X^2 + \epsilon, \epsilon \sim \mathcal{N}(0,1)$ and $\epsilon \Perp X$ in Figure~\ref{fig:dependence-map-simulated}, and for 250 observations from Frank copula $c_{\theta=5}^{{Frank}}$ in Figure~\ref{fig:dependence-map-simulated2}. 
Left figure shows the 250 observations. 
In the middle and right figures, red areas indicate estimated positive local dependence, light blue areas negative local dependence, respectively. Blue areas indicate low-density areas where estimations are not reliable. In Figure~\ref{fig:dependence-map-simulated}, both dependence maps capture the negative dependence for negative $X$ and the positive dependence for positive $X$. In Figure~\ref{fig:dependence-map-simulated2}, both dependence maps capture positive dependence in broad region. Overall, relative local dependence and local dependence do not exhibit differences in the dependence map.
\begin{figure}[htbp]
  \begin{minipage}[b]{0.25\linewidth}
    \centering
    \includegraphics[width=0.9\linewidth]{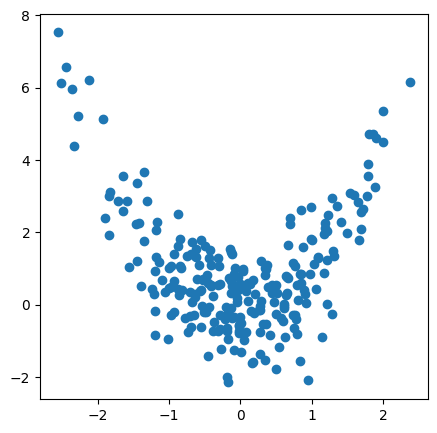}
  \end{minipage}
  \begin{minipage}[b]{0.25\linewidth}
    \centering
    \includegraphics[width=0.9\linewidth]{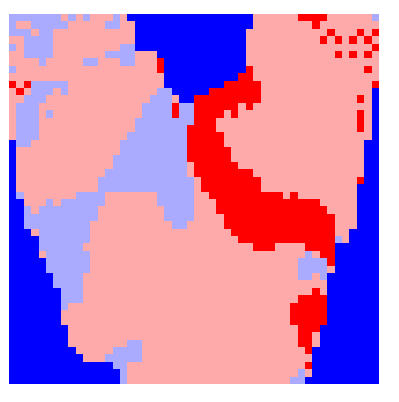}
  \end{minipage}
  \begin{minipage}[b]{0.25\linewidth}
    \centering
    \includegraphics[width=0.9\linewidth]{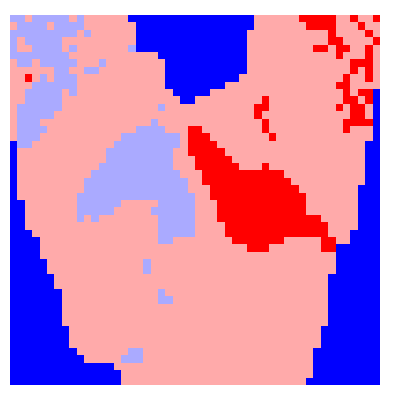}
  \end{minipage}
  \begin{minipage}[b]{0.2\linewidth}
    \centering
    \includegraphics[width=0.9\linewidth]{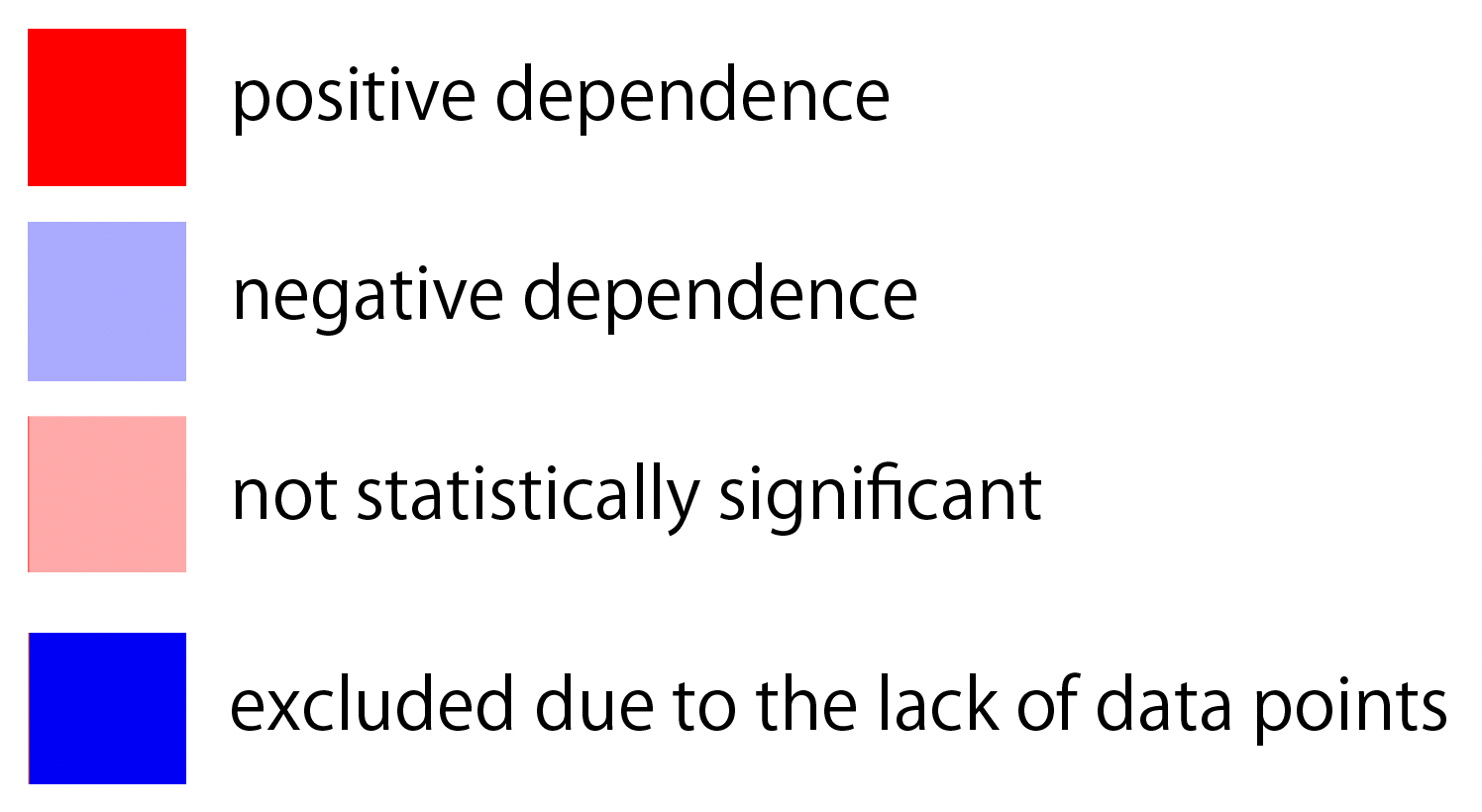}
    \vspace{1cm}
  \end{minipage}
  \caption{left : plot of observations / middle : dependence map of relative local dependence / right : dependence map of local dependence function}
  \label{fig:dependence-map-simulated}
\end{figure}

\begin{figure}[htbp]
  \begin{minipage}[b]{0.25\linewidth}
    \centering
    \includegraphics[width=\linewidth]{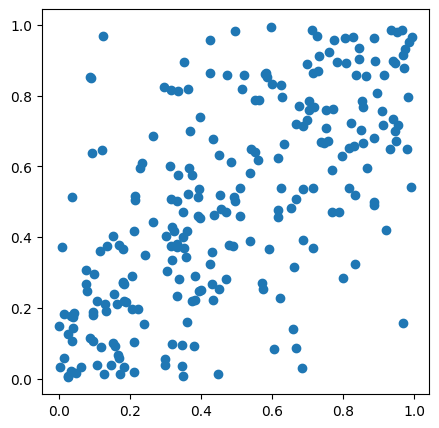}
  \end{minipage}
  \begin{minipage}[b]{0.25\linewidth}
    \centering
    \includegraphics[width=\linewidth]{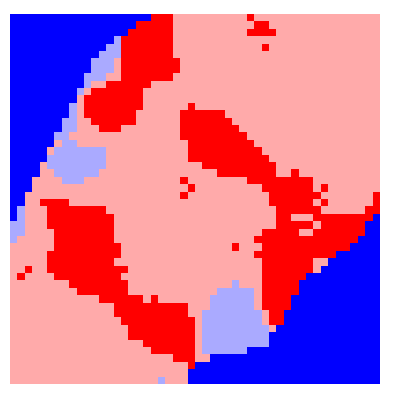}
  \end{minipage}
  \begin{minipage}[b]{0.25\linewidth}
    \centering
    \includegraphics[width=\linewidth]{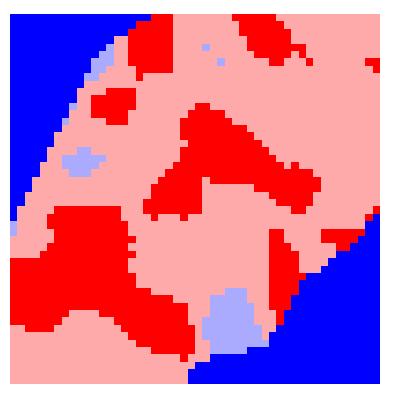}
  \end{minipage}
  \begin{minipage}[b]{0.2\linewidth}
    \centering
    \includegraphics[width=0.9\linewidth]{img/dependence_map/hanrei.png}
    \vspace{1cm}
  \end{minipage}
  \caption{left : plot of observations / middle : dependence map of relative local dependence / right : dependence map of local dependence function}
  \label{fig:dependence-map-simulated2}
\end{figure}

\subsection{Local Frank fitting}

In the kernel weighted local likelihood maximization estimation~\cite{hjort1996locally}, 
$$\frac{1}{n}\sum_{i=1}^n K_{h_1}(X_i-x)K_{h_2}(Y_i-y) \log {f(X_i,Y_i)} - \int \int K_{h_1}(X-x)K_{h_2}(Y-y)f(X,Y)\mathrm{d}X\mathrm{d}Y$$
is maximized with respect to the parameters included implicitly in $f$, where $K_h(\cdot) = \frac{1}{h}K(\frac{\cdot}{h})$.
Specifically in the local bilinear fitting~\cite{jones2003} used in Section 4.1., a bilinear form with four parameters is used to approximate $f$. By assuming the form $\log{f(x,y)} = a + bx + cy + dxy$, the estimator of the local dependence can be obtained naively as $\hat{d}$ because $ \frac{\partial^2}{\partial x \partial y}\log{f(x,y)} = d$.
Similarly, we propose the \textit{local Frank fitting} as the estimator specifically for the relative local dependence.
Here, we consider using Frank copula density with a single parameter $\theta$, instead. The local likelihood becomes
$$L_n(x,y,\theta;\{(x_i,y_i)\}) = \frac{1}{n}\sum_{i=1}^n K_{h_1}(x_i-x)K_{h_2}(y_i-y) \log{c_\theta^{Frank}(x_i,y_i)} - \int \int K_{h_1}(X-x)K_{h_2}(Y-y)c_\theta^{Frank}(X,Y)\mathrm{d}X\mathrm{d}Y$$
and the estimator of the relative local dependence at the point $(x,y)$ is calculated as $\hat{r} = 2\hat{\theta}$ where
\begin{equation} \label{eq:local-frank-fitting-estimator}
\hat{\theta} = \underset{\theta}{\mathrm{argmax}}\ L_n(x,y,\theta;\{(x_i,y_i)\})
\end{equation}
because $\frac{1}{f(x,y)}\frac{\partial^2}{\partial x \partial y}\log{f(x,y)} = \frac{1}{c_\theta^{Frank}(x,y)}\frac{\partial^2}{\partial x \partial y}\log{c_\theta^{Frank}(x,y)} = 2\theta$ when $f$ is truly the Frank density.

To evaluate the performance of this proposed estimator, we conducted the experiments using simulation data. First, we sampled data from the Frank copula with $\theta=3$ and from the Clayton copula with $\theta=5$, respectively. The number of samples in each case is set to 1000 or 10000. Then, to reduce computational time, we prepare $10 \times 10$ representative points at even intervals, $\{(x,y)\ |\ x,y = 0.15, \dots, 0.95\}$, excluding the borders. On each point, the relative local dependence is estimated by the naive estimator (based on the local bilinear fitting) and the local Frank fitting estimator.
Note that for the local Frank fitting estimator, since the maximization in \eqref{eq:local-frank-fitting-estimator} is not tractable, we employ the numerical integral method and the numerical search of $\theta$ among the candidates with the interval of 0.2, i.e.,  $\{\theta^*-4, \theta^*-3.8, \dots, \theta^*+3.8, \theta^*+4\}$, where $\theta^*$ denotes the true value. 

The estimated results are depicted in Figure~\ref{fig:rld2-frank-result} and Figure~\ref{fig:rld2-clayton-result}.
Table~\ref{tab:error} shows the performance of these estimators. The error is calculated as the squared sum of difference between the estimated value and the true value among the region where sample points are dense enough. 
Overall, it is observed that the local Frank fitting outperforms the naive estimation, and that using more samples improve the estimation performance.

\begin{figure}
    \centering
    \includegraphics[width=0.8\linewidth]{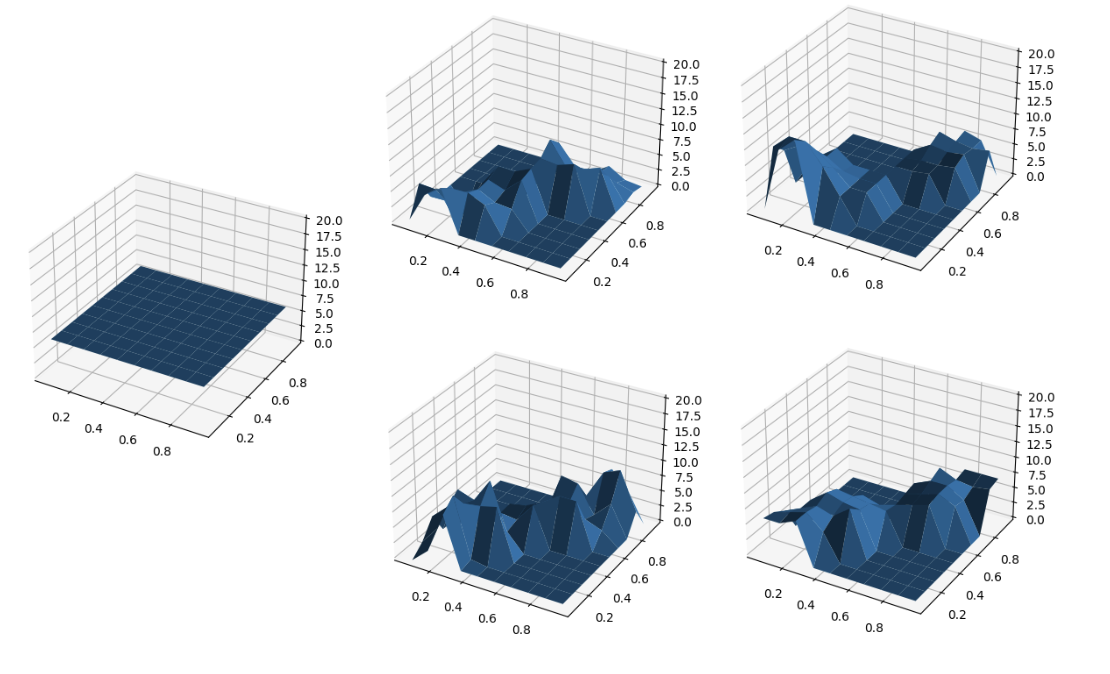}
    \caption{The estimation result of the relative local dependence of $c_{\theta=3}^{Frank}$. The left image shows the ground truth: $r(u,v) = 2\theta = 6$. The middle images show the naive estimations. The right images show the local Frank fitting. The top images are with 1000 samples, and the bottom images are with 10000 samples.}
    \label{fig:rld2-frank-result}

    \includegraphics[width=0.8\linewidth]{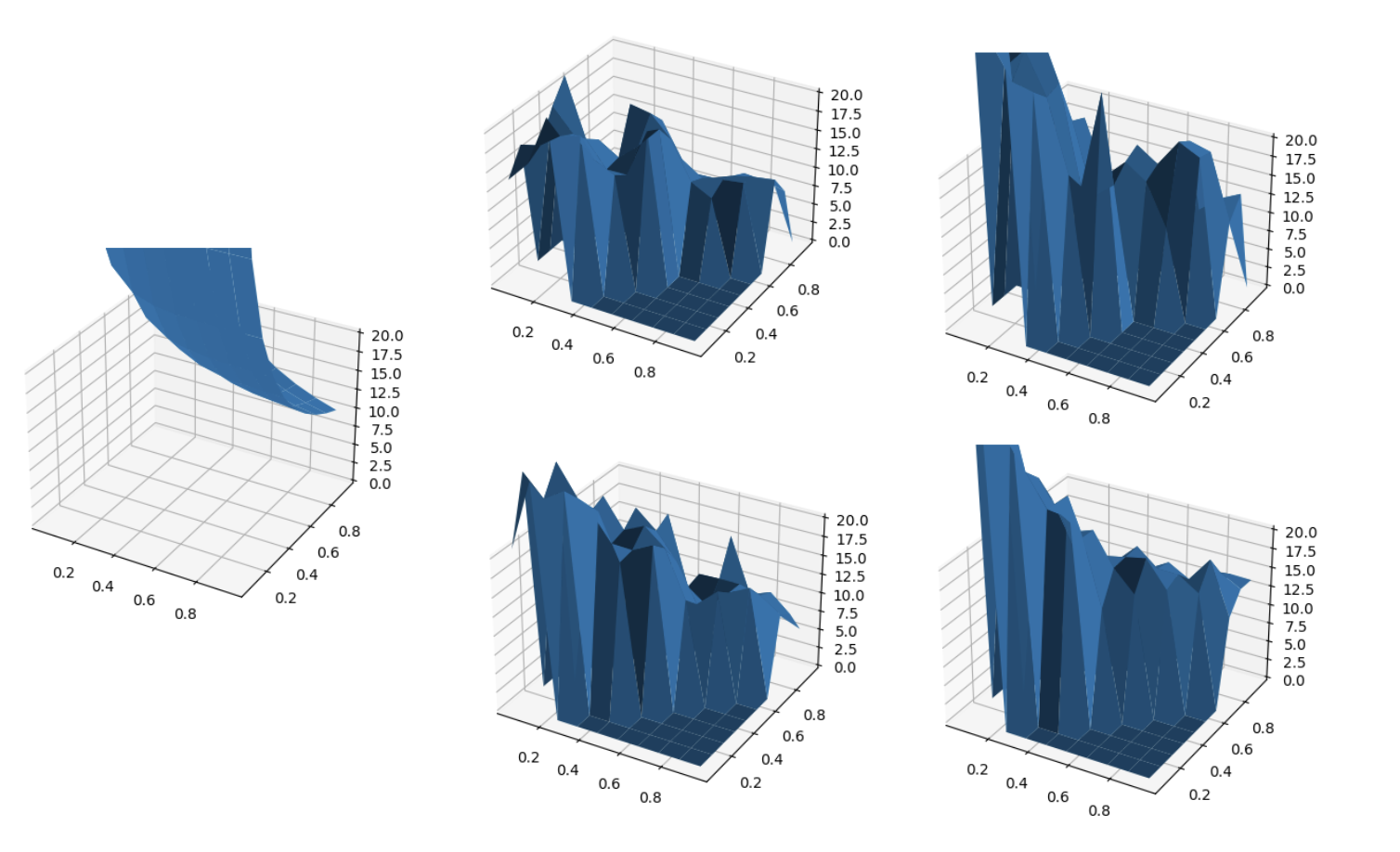}
    \caption{The estimation result of the relative local dependence of $c_{\theta=5}^{Clayton}$. The left image shows the ground truth: $r(u,v) = \frac{\theta(1+2\theta)}{\theta}C^{-1}(u,v)$. The middle images show the naive estimations. The right images show the local Frank fitting. The top images are with 1000 samples, and the bottom images are with 10000 samples.}
    \label{fig:rld2-clayton-result}
\end{figure}

\begin{table}
    \centering
    \begin{tabular}{cccc}
        True density & Estimation method & Sample size& Error \\ \hline
        Frank copula ($\theta=3$) &naive& 1000 & 22.16 \\ 
        Frank copula ($\theta=3$) &naive& 10000 & 18.91\\ 
        Frank copula ($\theta=3$) &local Frank fitting& 1000 &  \bf{14.58} \\ 
        Frank copula ($\theta=3$) &local Frank fitting& 10000 &  \bf{5.57}\\ \hline 
        Clayton copula ($\theta=5$) &naive & 1000 & 111.1 \\
        Clayton copula ($\theta=5$) &naive & 10000 & 74.60 \\
        Clayton copula ($\theta=5$) &local Frank fitting & 1000 & \bf{32.27} \\
        Clayton copula ($\theta=5$) &local Frank fitting& 10000 & \bf{19.51}\\
    \end{tabular}
    \caption{Estimation error of the estimators of relative local dependence.}
    \label{tab:error}
\end{table}


\section{New copula family defined by a partial differential equation}

In this section, we extend results of the previous section in a straightforward way so as to construct a new copula family. Inspired by the fact that Frank copula's relative local dependence is always constant, let us consider a copula $c_{k,\zeta}$ such that
\begin{equation}\label{eq:pde-k}
    \frac{1}{c(u,v)^{k}}\frac{\partial^2}{\partial u \partial v} \log{c(u,v)} = \zeta
\end{equation}
where $k$ is an integer and $\zeta\ (\neq 0)$ is a constant of proportionality. 
This family includes as typical examples uniform ($\zeta=0$), Frank ($k=1$), MICS ($k=0)$, and FGM ($k=-2$) copulas.

\subsection{Similarity between Frank and FGM copulas}

Each Frank copula and FGM copula have been among the typical choices when selecting copulas for tail-independent distributions. Their commonalities have never been mentioned, as far as the authors are aware. Here we utilize our characterization of these two copulas presented in \eqref{eq:pde-k} and discover a brand new similarity between them. First, we observe an invariant property of \eqref{eq:pde-k}.

\begin{proposition} \label{prop:k_zeta_Z}
    Assume a copula density $c_{k,\zeta}$ satisfies \eqref{eq:pde-k},
    where $k$ is a non-zero constant. 
    Then, the density function proportional to $c_{k,\zeta}(u,v)^k$ is a solution of the  PDE
    \begin{equation}\label{eq:prop_k_zeta_Z}
         \frac{\partial^2}{\partial u \partial v} \log{p(u,v)} = k\zeta Z p(u,v), Z = \int_0^1 \int_0^1 \{c_{k,\zeta}(u,v)\}^k dudv,
    \end{equation} 
    i.e., the Frank density $c_\theta^{\mathrm{Frank}}$ with the parameter $\theta = k\zeta Z/2$.
\end{proposition}
\begin{proof}
    $$\frac{\partial^2}{\partial u \partial v} \log{\{c_{k,\zeta}(u,v)\}} = \zeta \{c_{k,\zeta}(u,v)\}^k$$
    $$\frac{\partial^2}{\partial u \partial v} \log{[\{c_{k,\zeta}(u,v)\}^k]^{1/k}} = \zeta \{c_{k,\zeta}(u,v)\}^k$$
    $$\frac{1}{k}\frac{\partial^2}{\partial u \partial v} \log{\frac{\{c_{k,\zeta}(u,v)\}^k}{Z}} = \zeta \{c_{k,\zeta}(u,v)\}^k$$
    $$\frac{\partial^2}{\partial u \partial v} \log{\frac{\{c_{k,\zeta}(u,v)\}^k}{Z}} = k\zeta Z \frac{\{c_{k,\zeta}(u,v)\}^k}{Z}$$
\end{proof}


Proposition~\ref{prop:k_zeta_Z} provides insight into the copula family defined by \eqref{eq:pde-k}, indicating a scaling structure centered around the Frank copula. 
Note that a copula satisfying \eqref{eq:prop_k_zeta_Z} is known to be a Frank copula~\cite{sukeda2024frank}. 
Hence, we can say that after raised to the power of $k$ and then normalized, every copula in \eqref{eq:pde-k} becomes a probability density having a Frank copula.
One notable example is the FGM copula, which appears when $k=-2$. 
By simple calculations, the FGM copula density $c_\theta^{FGM}$ satisfies
$$\frac{\partial^2}{\partial u \partial v} \log{c_\theta^{FGM}(u,v)} = 4\theta c_\theta^{FGM}(u,v)^{-2}.$$
\noindent Therefore, the FGM density raised to the power of $-2$ followed by normalization has a Frank density. 
\begin{proposition}\label{prop:fgm-frank}
    Denote the density function of the FGM copula $c_\theta^{FGM}$. Then, the density function proportional to $(c_\theta^{FGM}(u,v))^{-2}$ has a Frank copula $c_{\theta'}^{Frank}$ with parameter $\theta' = -2\log{\frac{1+\theta}{1-\theta}}$.
\end{proposition}
\noindent In addition, since both densities of the FGM copula and Frank copula are explicitly known, this proposition can also be confirmed via direct calculations; see \ref{appendix:fgm-frank} for its proof. As far as the authors are aware, this relationship between the Frank copula and the FGM copula has not been mentioned in previous studies.

\subsection{Relationship with tail dependence}

Consider Archimedean copulas with generator function $\psi$. 
Assume generator $\psi$ is regularly varying with index $-\alpha$, which is denoted as $\psi \in RV_{-\alpha}$:
$$\lim_{x \to \infty} \frac{\psi(\lambda x)}{\psi(x)} = \lambda^{-\alpha}.$$
Larsson et al.~\cite{larsson2011extremal} showed that Archimedean copulas with generator $\psi \in RV_{-\alpha}$ has the tail dependence: $\lambda_L = \lim_{u\to0}\frac{C(u,u)}{u} = 2^{-\alpha}$.
Furthermore, Kurowicka and van Horssen~\cite{KUROWICKA2015127} pointed out that $\psi' \in RV_{-\alpha-1}$, $\psi'' \in RV_{-\alpha-2}$, $(\log{\psi''})'' \in RV_{-2}$, and that $i^c(u,u)$ is of order $u^{-2}$ as $u\to0$.
On the other hand, the relative local dependence $r^c(u,u)$ of Archimedean copulas with $\psi \in RV_{-\alpha}$ is of order $u^{-1}$ as $u\to0$.

\begin{proposition}
     If $\psi \in RV_{-\alpha}$ and the first and second derivatives of $\log{\psi''}$ are eventually monotone then relative local dependence $r^c(u,u)$ is of order $u^{-1}$ as $u\to0$.
\end{proposition}
\begin{proof}
We follow Section 2.3.4 of Kurowicka and van Horssen~\cite{KUROWICKA2015127}. 
    Since $\psi'' \in RV_{-\alpha-2}$ and $(\log{\psi''})'' \in RV_{-2}$, $\frac{1}{\psi''}(\log{\psi''})'' \in RV_{\alpha}$. On the other hand, $\psi^{-1} \in RV_{-1/\alpha}$. Hence, 
    $$r^c(u,u) = \left.\frac{1}{\psi''(t)}(\log{\psi''(t)})'' \right|_{t = 2\psi^{-1}(u)}$$
    is of order $u^{-1}$ as $u\to0$.
\end{proof}

\begin{example}[Clayton copula]
The generator of Clayton copula $\psi^{Clayton}$ is  $RV_{-1/\theta}$ because
$$\frac{\psi^{Clayton}(\lambda t)}{\psi^{Clayton}(t)} = (\frac{1+\lambda t}{1+t})^{-1/\theta} \to \lambda^{-1/\theta} (t \to \infty).$$
The density function on the corner becomes
\begin{align*}
c^{\mathrm{Clayton}}(u,u) &= \frac{1+\theta}{u^{2+2\theta}}\frac{1}{(2u^{-\theta}-1)^{\frac{1}{\theta}+2}}\\
&\to \frac{1+\theta}{2^{\frac{1}{\theta}+2}}u^{-1} \ (u\to 0)
\end{align*}
and the relative local dependence when approaching the left bottom corner $(0,0)$ is 
\begin{align*}
    r^{\mathrm{Clayton}}(u,u) &= \frac{\theta(1+2\theta)}{1+\theta} \{C^{\mathrm{Clayton}}(u,u)\}^{-1}\\
    &= \frac{\theta(1+2\theta)}{1+\theta}(2u^{-\theta}-1)^{1/\theta}\\
    &\to \frac{\theta(1+2\theta)}{1+\theta} 2^{1/\theta}u^{-1} \ (u\to 0)
\end{align*}

\end{example}


\subsection{Checkerboard copula of $c_{k,\zeta}$}

While the theoretical relationship presented in \eqref{eq:pde-k} is simple, explicit solutions are often not feasible, which is a significant limitation especially in practical use of this copula. 
A typical approach to overcome the issue of intractable copulas is through discrete approximation, often called checkerboard copulas. Considering the checkerboard copula obtained by dividing the interval $[0, 1]$ into equally spaced intervals, it essentially becomes a contingency table problem, or a problem of matrices in other words.
Fortunately, the checkerboard approximation of these copula densities becomes tractable by Algorithm 1. 
Algorithm 1 is a variant of Sukeda and Sei~\cite{sukeda2023minimum} with a slight modification in line 4, which fixes the entry of each checkerboard accordingly to align the value of local dependence without violating the copula constraints, i.e, row sums and column sums are always $\frac{1}{n}$ when $n$ denotes the gridsize of the checkerboard. 
Note that line 4 is just solving the following quadratic equation with respect to $\delta$,
$$(1-T)\delta^2 + \{\pi_{ij} + \pi_{i+1,j+1} + (\pi_{i+1,j} + \pi_{i,j+1})T\}\delta + \pi_{ij}\pi_{i+1,j+1}-\pi_{i+1,j}\pi_{i,j+1}T = 0$$
where $T = \exp{(\zeta(\pi_{ij}+\pi_{i+1,j}+\pi_{i+1,j}+\pi_{i+1,j+1})^k)}$, 
thus easily handled as well. 

\begin{algorithm}[H]
    \caption{Greedy calculation}
    \label{alg1}
    \begin{algorithmic}[1]
    \REQUIRE constant $\zeta$, parameter $k$
    \STATE $\Pi \leftarrow n\times n$ uniform copula
    \WHILE{converge}
    \STATE Choose a $2 \times 2$ submatrix  $\begin{pmatrix}
\pi_{i,j}&\pi_{i, j+1}\\
\pi_{i+1, j}&\pi_{i+1, j+1}\\
\end{pmatrix}$ of $\Pi$.
    \STATE Solve $$\frac{1}{(\pi_{ij}+\pi_{i+1,j}+\pi_{i,j+1}+\pi_{i+1,j+1})^k}\log{\frac{(\pi_{i,j}+\delta)(\pi_{i+1,j+1}+\delta)}{(\pi_{i+1,j}-\delta)(\pi_{i,j+1}-\delta)}} = \zeta$$
    \STATE Update $\begin{pmatrix}
\pi_{i,j}&\pi_{i,j+1}\\
\pi_{i+1,j}&\pi_{i+1,j+1}\\
\end{pmatrix}\to 
\begin{pmatrix}
\pi_{i,j}+\delta&\pi_{i,j+1}-\delta\\
\pi_{i+1,j}-\delta&\pi_{i+1,j+1}+\delta\\
\end{pmatrix}$
    \ENDWHILE
    \end{algorithmic}
\end{algorithm}

Figure~\ref{fig:5x5} displays the behavior of $c_{k,\zeta}$ when varying $k$ and $\zeta$ under $n=5$, while the bottom right part is omitted because of the computational issue. It can be observed that the density in tails and diagonals grow when $\zeta$ gets larger and $k$ gets smaller. 

\begin{figure}
    \centering
    \includegraphics[width=1\linewidth]{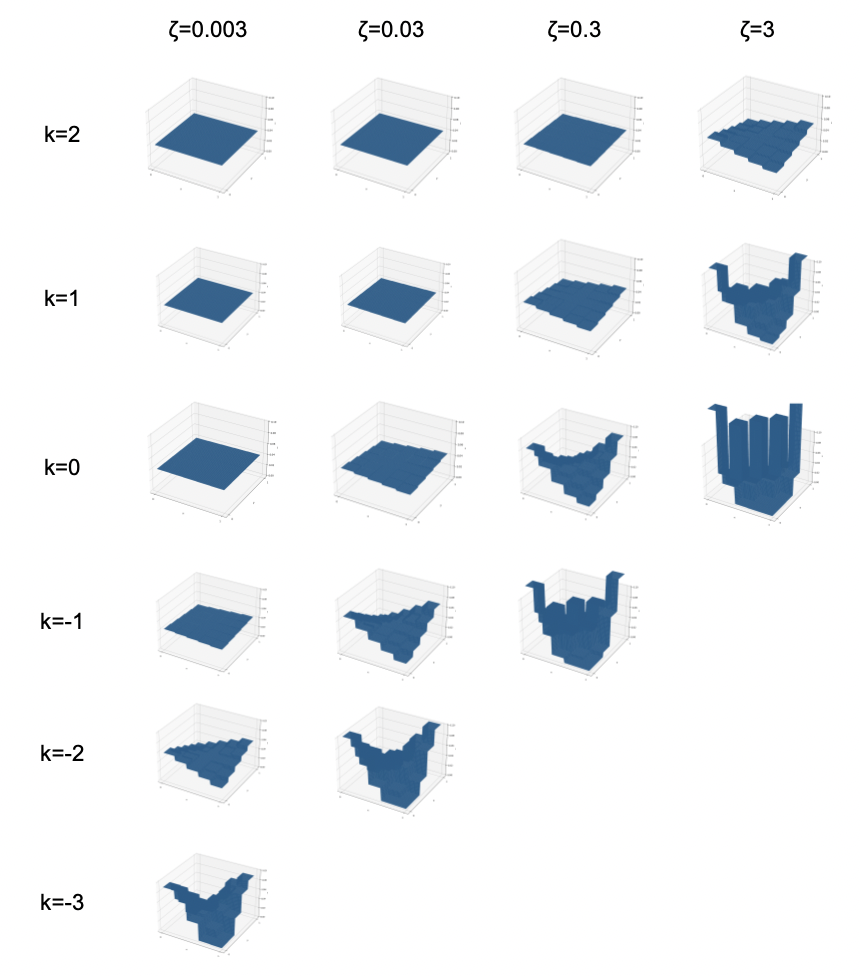}
    \caption{$c_{k,\zeta}$ when varying $k$ and $\zeta$}
    \label{fig:5x5}
\end{figure}

\section{Conclusion}

In this paper we proposed the use of relative local dependence, and studied its properties and the estimation method of it for typical copulas including the Archimedean copulas. Different from the existing local dependence function, the relative local dependence is invariant to monotone transformation in marginal distributions, which is often considered preferable as a notion of dependence. 
Since every Archimedean copula is defined by a single generator function in a simple manner, we show that a part of them can be characterised by the relationship between its copula density and its relative local dependence. Specifically, the Frank copula exhibits a constant relative local dependence everywhere on a unit square. On the other hand, the constant local dependence function leads to a different copula known as the minimum information copula.
This new characterization reveals a non-obvious relationship between the Frank copula and the FGM copula. Moreover, it can be further extended to provide new copulas defined as a solution of the partial differential equation, which states that the local dependence function is proportional to its copula density to the power of $k$. These copulas becomes tractable in practice due to the greedy numerical algorithm, while their explicit formulations are not known.

The estimation procedure of relative local dependence, conducted in Section 3.5, could be further improved. Mitigating the weakness of the kernel-based method, suffering on a strong bias on the border of the support, is an important future work.
Last but not least, we only take bivariate copulas into account in this work. All the relationship between the original density function and the (relative) local dependence presented in this paper does not necessarily hold in higher dimension, even for trivariate copulas. For the local dependence function, Kurowicka and van Horssen~\cite{KUROWICKA2015127} has investigated high-dimensional Archimedean copulas. Characterising the relative local dependence for multivariate copulas in a consistent way is also a future work.

\section*{Declaration of the use of generative AI and AI-assisted technologies}

During the preparation of this work the authors used ChatGPT-3.5 in order to correct the English text. After using this tool/service the authors reviewed and edited the content as necessary and takes full responsibility for the content of the publication.

\section*{Acknowledgments}

Issey Sukeda is supported by RIKEN Junior Research Associate Program. 
Tomonari Sei is supported by JSPS KAKENHI (19K11865, 21K11781).

\bibliographystyle{abbrv}
\bibliography{References}

\appendix

\section{Relative local dependence of a certain type of Archimedean copulas}

Here let us consider the Archimedean copula whose generator function has the form of $\psi(t) = \exp{(A(t))}$; $A(t) = -\frac{1}{\theta}\log{(1+t)}$ for Clayton copula and $A(t) = -t^{1/\theta}$ for GH copula, for instance. In this case, the derivatives of $\psi$ can be expressed via $\psi$:
$\psi'(t) = A'(t)\psi(t)$, $\psi''(t) = \{A''(t)+(A'(t))^2\}\psi(t)$, and $\log{\psi''(t)} = \log{\{A''(t)+(A'(t))^2\}} + \log{\psi(t)}$.
Hence, the relative local dependence is calculated as 
\begin{align*}
    r(u,v) &= \left. \frac{1}{\psi''(t)}\frac{\partial^2}{\partial t^2}\log{\psi''(t)}\right|_{t=\psi^{-1}(u)+\psi^{-1}(v)} \\
    &= \left. \frac{1}{\psi(t)} \left(\frac{B''(t)B(t)-B'(t)^2}{B(t)^2}+A''(t) \right) \right|_{t=\psi^{-1}(u)+\psi^{-1}(v)} \ (B(t) = A^{(2)}(t) + \{A^{(1)}(t)\}^2)\\
    &= \frac{A^{(4)}A^{(2)} + A^{(4)}\{A^{(1)}\}^2 + \{A^{(1)}\}^4 A^{(2)} + 3\{A^{(2)}\}^3 + 2\{A^{(1)}\}^3A^{(3)}-2A^{(1)}A^{(2)}A^{(3)} - \{A^{(3)}\}^2}{(A^{(2)} + \{A^{(1)}\}^2)^3} \frac{1}{C(u,v)}
\end{align*}
where $A^{(i)}$ denotes the $i$-th derivative of $A(t)$ with respect to $t$.

\section{Proof of Proposition~\ref{prop:tau_LL}} \label{appendix:local-kendall-tau}
\begin{proof}   
   Here we introduce the notations
   $$\begin{cases}
       A_1&=\int_{\{(u,\tilde{u},v,\tilde{v})\ |\ 0 \leq  \tilde{u} \leq u \leq p,\ 0 \leq \tilde{v} \leq v \leq q\}} c(u,v)c(\tilde{u},\tilde{v})dudvd\tilde{u}d\tilde{v}\\
       A_2&=\int_{\{(u,\tilde{u},v,\tilde{v})\ |\ 0 \leq u \leq \tilde{u}  \leq p,\ 0 \leq \tilde{v} \leq v \leq q\}} c(u,v)c(\tilde{u},\tilde{v})dudvd\tilde{u}d\tilde{v}\\
       A_3&=\int_{\{(u,\tilde{u},v,\tilde{v})\ |\ 0 \leq  \tilde{u} \leq u \leq p,\ 0 \leq v \leq \tilde{v} \leq q\}} c(u,v)c(\tilde{u},\tilde{v})dudvd\tilde{u}d\tilde{v}\\
       A_4&=\int_{\{(u,\tilde{u},v,\tilde{v})\ |\ 0 \leq u \leq  \tilde{u}  \leq p,\ 0 \leq v \leq \tilde{v}  \leq q\}} c(u,v)c(\tilde{u},\tilde{v})dudvd\tilde{u}d\tilde{v}\\
   \end{cases}.$$
   Using these notations,
    \begin{align*}
        C(p,q)^2 &= (\int_{0}^{p} \int_{0}^{q} c(u,v) dudv)(\int_{0}^{p}\int_{0}^{q} c(\tilde{u},\tilde{v})d\tilde{u}d\tilde{v})\\
        &= \int_{0}^{p} \int_{0}^{q}\int_{0}^{p}\int_{0}^{q} c(u,v)c(\tilde{u},\tilde{v})dudvd\tilde{u}d\tilde{v}\\
        &= \int_{\{(u,\tilde{u},v,\tilde{v})\ |\ 0 \leq  \tilde{u}, u \leq p,\ 0 \leq \tilde{v}, v \leq q\}} c(u,v)c(\tilde{u},\tilde{v})dudvd\tilde{u}d\tilde{v}\\
        &= A_1 + A_2 + A_3 + A_4.
    \end{align*}
    Due to symmetry between $(u,v)$ and $(\tilde{u},\tilde{v})$, $A_1 = A_4$ and $A_2 = A_3 = \frac{1}{2}\{C(p,q)\}^2 - A_1$. 
    Hence, 
    \begin{align*}
    \int_{0}^{p} \int_{0}^{q}\int_{0}^{p}\int_{0}^{q} \mathrm{sgn}(u-\tilde{u})\mathrm{sgn}(v-\tilde{v})c(u,v)c(\tilde{u},\tilde{v})dudvd\tilde{u}d\tilde{v}
    &= A_1 - A_2 - A_3 + A_4 \\
    &= 4A_1 - C(p,q)^2 
    \end{align*}
    Therefore, 
    $$\tau_{LL}=\frac{4\int_0^p \int_0^q C(u,v)dC}{C(p,q)^2}- 1 = \frac{4A_1-C(p,q)^2}{C(p,q)^2} = \frac{\int_{0}^{p} \int_{0}^{q}\int_{0}^{p}\int_{0}^{q} \mathrm{sgn}(u-\tilde{u})\mathrm{sgn}(v-\tilde{v})c(u,v)c(\tilde{u},\tilde{v})dudvd\tilde{u}d\tilde{v}}{C(p,q)^2}$$
\end{proof}

\section{Proof of Proposition~\ref{prop:fgm-frank}} \label{appendix:fgm-frank}
\begin{proof}
Let $f(u,v)$ denote the target density:
$$f(u,v) = \frac{(c_\theta^{FGM}(u,v))^{-2}}{Z},$$
    where $Z$ denotes the normalization factor. $Z$ can be explicitly calculated as 
    \begin{equation}
    Z = \int_0^1 \int_0^1 (c_\theta^{FGM}(u,v))^{-2} du dv = \frac{1}{2\theta}\log{\frac{1+\theta}{1-\theta}}.\label{eq:FGMnorm}
    \end{equation}
    To calculate the copula of $f(u,v)$, we need the marginal functions $f_x(x)$, $f_y(y)$, $F_x(x)$, and $F_y(y)$. These calculations are straightforward:
    $$f_x(x) = \int_0^1 dy f(x,y) = \frac{1}{Z}\frac{1}{1+\theta(2x-1)}\frac{1}{1-\theta(2x-1)} = \frac{1}{2Z}(\frac{1}{1+\theta(2x-1)}+\frac{1}{1-\theta(2x-1)}),$$
    \begin{equation}
        F_x(x) = \int_0^x f_x(x) dx = \frac{1}{4\theta Z}(\log{\frac{1+\theta}{1-\theta}}+\log{\frac{1+\theta(2x-1)}{1-\theta(2x-1)}}).\label{eq:FGMcdf}
    \end{equation}

    The copula density associated with $f(x,y)$ is calculated as 
    \begin{align*}
        c(u,v) &= \frac{f(F_x^{-1}(u),F_y^{-1}(v))}{f_x(F_x^{-1}(u))f_y(F_y^{-1}(v))}\\
        &= Z\frac{(1+\theta(2F_x^{-1}(u)-1))(1-\theta(2F_x^{-1}(u)-1))(1+\theta(2F_y^{-1}(v)-1))(1-\theta(2F_y^{-1}(v)-1))}{(1+\theta(2F_x^{-1}(u)-1)(2F_y^{-1}(v)-1))^2}\\
        &= Z(\frac{4\theta}{1-\theta})^2\frac{(1+\theta(2F_x^{-1}(u)-1))(1-\theta(2F_x^{-1}(u)-1))(1+\theta(2F_y^{-1}(v)-1))(1-\theta(2F_y^{-1}(v)-1))}{(\frac{4\theta}{1-\theta}+\frac{4}{1-\theta}\theta^2(2F_x^{-1}(u)-1)(2F_y^{-1}(v)-1))^2}\\
        &= Z(\frac{4\theta}{1-\theta})^2\frac{(1+\theta(2F_x^{-1}(u)-1))(1-\theta(2F_x^{-1}(u)-1))(1+\theta(2F_y^{-1}(v)-1))(1-\theta(2F_y^{-1}(v)-1))}{(\frac{1+\theta}{1-\theta}(2+2\theta^2(2F_x^{-1}(u)-1)(2F_y^{-1}(v)-1))-(2-2\theta^2(2F_x^{-1}(u)-1)(2F_y^{-1}(v)-1)))^2}
    \end{align*}
    We denote $H_x^+(u) = 1+\theta(2F_x^{-1}(u)-1), H_x^-(u) = 1-\theta(2F_x^{-1}(u)-1), H_y^+(v) = 1+\theta(2F_y^{-1}(v)-1), H_y^-(v) = 1-\theta(2F_y^{-1}(v)-1)$ for convenience. Then,
    \begin{align*}
        c(u,v) &= Z(\frac{4\theta}{1-\theta})^2 \frac{H_x^+(u)H_x^-(u)H_y^+(v)H_y^-(v)}{\{\frac{1+\theta}{1-\theta} (H_x^+(u)H_y^+(v) + H_x^-(u)H_y^-(v)) - (H_x^+(u)H_y^-(v) + H_x^-(u)H_y^+(v))\}^2}\\
        &= Z(\frac{4\theta}{1-\theta})^2 \frac{\frac{H_x^+(u)}{H_x^-(u)}\frac{H_y^+(v)}{H_y^-(v)}}{\{\frac{1+\theta}{1-\theta} \frac{H_x^+(u)}{H_x^-(u)}\frac{H_y^+(v)}{H_y^-(v)} + \frac{1+\theta}{1-\theta} - \frac{H_x^+(u)}{H_x^-(u)} - \frac{H_y^+(v)}{H_y^-(v)}\}^2}\\
        &= Z(\frac{4\theta}{1-\theta})^2 \frac{\frac{1+\theta}{1-\theta}\frac{H_x^+(u)}{H_x^-(u)}\frac{1+\theta}{1-\theta}\frac{H_y^+(v)}{H_y^-(v)}}{\{\frac{1+\theta}{1-\theta} \frac{H_x^+(u)}{H_x^-(u)}\frac{1+\theta}{1-\theta}\frac{H_y^+(v)}{H_y^-(v)} + (\frac{1+\theta}{1-\theta})^2 - \frac{1+\theta}{1-\theta}\frac{H_x^+(u)}{H_x^-(u)} - \frac{1+\theta}{1-\theta}\frac{H_y^+(v)}{H_y^-(v)}\}^2}\\
    \end{align*}
    Note that by applying $x = F_x^{-1}(u)$ to \eqref{eq:FGMcdf}, we obtain
    $$\exp{(4\theta Zu)} = \frac{1+\theta}{1-\theta}\frac{1+\theta(2F_x^{-1}(u)-1)}{1-\theta(2F_x^{-1}(u)-1)} = \frac{1+\theta}{1-\theta}\frac{H_x^+(u)}{H_x^-(u)}.$$
    From \eqref{eq:FGMnorm}, we also have 
    $$\exp{(4\theta Z)} = (\frac{1+\theta}{1-\theta})^2.$$
    Therefore, the copula density of interest becomes
    \begin{align*}
        c(u,v) &= Z(\frac{4\theta}{1-\theta})^2 \frac{\exp{(4\theta Zu)}\exp{(4\theta Zv)}}{\{\exp{(4\theta Zu)}\exp{(4\theta Zv)} + \exp{(4\theta Z)} - \exp{(4\theta Zu)} -\exp{(4\theta Zv)}\}^2}\\
        &= 4\theta Z(\exp{(4\theta Z)}-1) \frac{\exp{(4\theta Zu)}\exp{(4\theta Zv)}}{\{1-\exp{(4\theta Z)} - (1-\exp{(4\theta Zu)})(1-\exp{(4\theta Zu)})\}^2}\\
        &= \frac{-4\theta Z(1-\exp{(4\theta Z)}) \exp{(4\theta Z(u+v))}}{\{1-\exp{(4\theta Z)} - (1-\exp{(4\theta Zu)})(1-\exp{(4\theta Zu)})\}^2}.\\
    \end{align*}
    Here, we see that this is a Frank density. By \eqref{eq:FGMnorm}, the parameter of this Frank copula is $-4\theta Z = -2\log{\frac{1+\theta}{1-\theta}}.$
\end{proof}

\end{document}